\documentclass[11pt]{article}
\usepackage{times}
\usepackage{latexsym,amsmath,amssymb,url,amsthm}
\usepackage[totalwidth=15.0cm,totalheight=20.0cm]{geometry}

\newtheorem{lemma}{Lemma}
\newtheorem{corollary}{Corollary}
\newtheorem{definition}{Definition}
\newtheorem{theorem}{Theorem}
\newtheorem{proposition}{Proposition}

\newtheorem{assump}{Assumption}
\newtheorem{remark}{Remark}

\newcommand{\pp}{\hat{\phi}}

\newcommand{\sat}{\ensuremath{ {\rm sat}}}
\newcommand{\unsat}{\ensuremath{{\rm unsat}}}
\newcommand{\E}{\ensuremath{\mathbb{E}}}

\newcommand{\newhard}{F_h}
\newcommand{\oldhard}{{\cal C}_h(F)}
\newcommand{\newsoft}{F_s}
\newcommand{\oldsoft}{{\cal C}(F) \setminus {\cal C}_h(F)}
\newcommand{\oldunitvar}{V_u(F)}
\newcommand{\newunitvar}{V_1}
\newcommand{\oldtwovar}{V_{h \setminus u}(F)}
\newcommand{\newtwovar}{V_2}
\newcommand{\oldsoftvar}{V_{\bar{h}}(F)}
\newcommand{\newsoftvar}{V_s}

\newif\ifdominik
\dominiktrue

\newcommand{\hard}{\ifdominik \newhard \else \oldhard \fi}
\newcommand{\soft}{\ifdominik \newsoft \else \oldsoft \fi}
\newcommand{\unitvar}{\ifdominik \newunitvar \else \oldunitvar \fi}
\newcommand{\twovar}{\ifdominik \newtwovar \else \oldtwovar \fi}
\newcommand{\softvar}{\ifdominik \newsoftvar \else \oldsoftvar \fi}
\newcommand{\hardvars}{ |\unitvar| + |\twovar|}

\begin{document}

\title{A New Bound for $3$-Satisfiable MaxSat and its Algorithmic Application\footnote{A preliminary version of this paper has appeared in \cite{GJY}. The most significant difference between this version and \cite{GJY} is a new proof of Theorem \ref{main_bound}, which is shorter and simpler than the proof in \cite{GJY}.}}

\author
{
G. Gutin, M. Jones\\
{\small Department of Computer Science}\\[-3pt]
{\small Royal Holloway, University of London}\\[-3pt]
{\small Egham, Surrey, TW20 0EX, UK}\\[-3pt]
{\small \url{{gutin|markj}@cs.rhul.ac.uk}}
\and D. Scheder\\
{\small Department of Computer Science}\\[-3pt]
{\small Aarhus University}\\[-3pt]
{\small DK-8000 Aarhus C, Denmark}\\[-3pt]
{\small \url{dominik.scheder@gmail.com}}
\and  A. Yeo\\
{\small Department of Mathematics}\\[-3pt]
{\small University of Johannesburg}\\[-3pt]
{\small Auckland Park, 2006 South Africa}\\[-3pt]
{\small \url{anders.yeo.work@gmail.com}}
}

\date{}
\maketitle

\begin{abstract}

Let $F$ be a CNF formula with $n$ variables and $m$ clauses. $F$ is 3-satisfiable if for any 3 clauses in $F$,
there is a truth assignment which satisfies all of them. Lieberherr and Specker (1982) and, later, Yannakakis  (1994) proved
that in each 3-satisfiable CNF formula at least $\frac{2}{3}$ of its clauses can be satisfied by a truth assignment.
We improve this result by showing that every 3-satisfiable CNF formula $F$ contains a subset of variables $U$,
such that some truth assignment $\tau$ will satisfy at least $\frac{2}{3}m+\frac{1}{3}m_U+\rho n'$ clauses, where $m$ is the number of
clauses of $F$, $m_U$ is the number of clauses of $F$ containing a variable from $U$, $n'$ is the total number of variables in
clauses not containing a variable in $U$, and $\rho$ is a positive absolute constant. Both $U$ and $\tau$
can be found in polynomial time.

We use our result to show that the following parameterized problem is fixed-parameter tractable and, moreover, has a kernel with a linear number of variables. In {\sc $3$-S-MaxSat-AE}, we are given a $3$-satisfiable CNF formula $F$ with $m$ clauses and asked to determine whether there is an assignment which satisfies at least $\frac{2}{3}m + k$ clauses, where $k$ is the parameter.
\end{abstract}

\pagenumbering{arabic}
\pagestyle{plain}

\section{Introduction}

We consider a formula $F$ in conjunctive normal form (CNF) as a set of clauses: $F=\{C_1,\ldots, C_m\}$. Each clause $C_i$ has an associated positive integral weight $w(C_i)$, and we let $w(F)$ denote the total weight of clauses of $F$.
(We use weighted clauses rather than letting $F$ be a multiset of clauses; see a remark in the end of this section for a discussion on the matter.)
In what follows, we assume that no clause contains both a variable and its negation, and no clause is empty. The set of variables of $F$ will be denoted by $V(F)$.
For a truth assignment $\tau$, let ${\rm sat}_{\tau}(F)$ be the total weight of clauses of $F$ satisfied by $\tau$ and let
${\rm sat}(F)$ be the maximum total weight of clauses of $F$ that can be satisfied by a truth assignment.

For any integer $t$, we say $F$ is \emph{$t$-satisfiable} if for any $t$ clauses in $F$ there exists a truth
assignment that satisfies all of them. Thus, every CNF formula $F$ is 1-satisfiable;
if $F$ is $2$-satisfiable then $F$ contains no pair of clauses of
the form $\{x\}, \{\bar{x}\}$;  if $F$ is $3$-satisfiable then the forbidden sets of clauses are pairs of
the form $\{x\}, \{\bar{x}\}$ and triplets of the form $\{x\}, \{y\},
\{\bar{x}, \bar{y}\}$ or $\{x\}, \{\bar{x}, y\}, \{\bar{x}, \bar{y}\}$, as
well as any triplets that can be derived from these by switching positive
literals with negative literals.

It is well-known that for any $1$-satisfiable CNF formula $F$, ${\rm sat}(F) \ge \frac{1}{2}w(F)$. Lieberherr and Specker \cite{LieberherrSpecker81,LieberherrSpecker82} and, later, Yannakakis  \cite{Yannakakis94} proved the following:
if $F$ is $2$-satisfiable then ${\rm sat}(F) \ge \pp w(F)$ (where $\pp \approx 0.61803$ is the positive root of $x^2+x = 1$);
if $F$ is $3$-satisfiable then ${\rm sat}(F) \ge \frac{2}{3}w(F)$.
These bounds are asymptotically tight (that is, for any $\epsilon > 0$, there exists a $3$-satisfiable CNF formula $F$ such that ${\rm sat}(F) < (\frac{2}{3} + \epsilon)w(F)$  and similar inequalities hold for $1$-satisfiable and  $2$-satisfiable formulas).

Crowston et al. \cite{CroGutJonYeo} strengthened the bound ${\rm sat}(F) \ge \pp w(F)$ for $2$-satisfiable CNF formulas to ${\rm sat}(F) \ge \pp w(F) + \gamma |V(F)|$ (where $\gamma \approx 0.072949$) using deterministic combinatorial arguments. In this paper, we strengthen the bound ${\rm sat}(F) \ge \frac{2}{3}w(F)$ for $3$-satisfiable CNF formulas. The deterministic approach of Crowston et al. \cite{CroGutJonYeo} cannot be readily extended to the 3-satisfiability case (which appears to be more complicated) and we use probabilistic arguments instead.

Our main results on $3$-satisfiable CNF formulas are as follows. A CNF formula $F$ is {\em expanding} if for every subset $X$ of the variables of $F$, the total weight of clauses containing variables of $X$ is not smaller than $|X|.$ We show that there is a positive absolute constant $\rho$ such that for every expanding $3$-satisfiable CNF formula $F$, we have \begin{equation}\label{expeq}{\rm sat}(F) \ge \frac{2}{3}w(F)+\rho|V(F)|.\end{equation} Using (\ref{expeq}) and a result on autarkies (defined in the next section) we obtain that there is a positive absolute constant $\rho$ such that for every $3$-satisfiable CNF formula $F$ we can find, in polynomial time, a subset $U$ of $V(F)$ and a truth assignment $\tau$ for which
\begin{equation}\label{3eq}{\rm sat}_{\tau}(F)\ge \frac{2}{3}w(F)+\frac{1}{3}w(F_U)+\rho |V(F\setminus F_U)|,\end{equation}
where $F_U$ is the subset of $F$ consisting of all clauses with a variable of $U$.
Note that (\ref{3eq}) improves the bound ${\rm sat}(F) \ge \frac{2}{3}w(F)$ for $3$-satisfiable formulas. Bound (\ref{3eq}) has an application in parameterized
algorithmics as described below.

Mahajan and Raman \cite{Mahajan99} considered the following parameterized problem {\sc SAT-AE}\footnote{AE stands for
Above Expectation}: we are given a
($1$-satisfiable) CNF formula $F$ and asked to determine whether there is an assignment which satisfies at least
$\frac{1}{2}w(F) + k$ clauses, where $k$ is the parameter.
(Basic notions on parameterized algorithms and complexity are given in Section \ref{sec:n}.)
For {\sc SAT-AE}, Mahajan and Raman \cite{Mahajan99} obtained a kernel with at most $6k+3$ variables and $10k$ clauses.
Crowston et al. \cite{CroGutJonYeo} improved this to $4k$ variables and $(2\sqrt{5}+4)k$ clauses.

As mentioned above Crowston et al. \cite{CroGutJonYeo} obtained the bound ${\rm sat}(F) \ge \pp w(F) + \gamma |V(F)|$ for $2$-satisfiable CNF formula $F$. This bound allowed them to solve an open problem of Mahajan and Raman \cite{Mahajan99} by proving
that the following parameterized problem is fixed-parameter tractable and, moreover, has a kernel with a linear number of variables. In {\sc $2$-S-MaxSat-AE}, we are given a $2$-satisfiable CNF formula $F$ and asked to determine whether there is an assignment which satisfies at least $\pp w(F) + k$ clauses, where $k$ is the parameter.

Bound (\ref{3eq}) allows us to prove that the following parameterized problem is fixed-parameter tractable and, moreover, has a kernel with a linear number of variables. In {\sc $3$-S-MaxSat-AE}, we are given a $3$-satisfiable CNF formula $F$ and asked to determine whether there is an assignment which satisfies at least $\frac{2}{3} w(F) + k$ clauses, where $k$ is the parameter. This answers a question from \cite{CroGutJonYeo}.

A parameterization of {\sc Max-$r$-SAT} above a tight lower bound was recently studied in \cite{alonALG,CroGutJon,CroGutJonKimRuz,KimWil}. Approaches used there are completely different from the one used in this paper.

Our paper is organized as follows.
In Section \ref{sec:n}, we provide additional terminology and notation.
In Section \ref{sec:mr}, we describe main results of the paper and also prove that {\sc $3$-S-MaxSat-AE} is fixed parameter tractable, and has a kernel with a linear number of variables. In the next two sections, we prove our main techical results that imply (\ref{expeq}) and (\ref{3eq}).
Finally, in Section \ref{sec:d} we state two open problems on $t$-satisfiable CNF formulas for any $t.$

\vspace{3mm}

\begin{remark}
{\rm Instead of assuming $F$ to be a set of clauses and having integral weights on the clauses, we could have allowed $F$ to be a multiset and the clauses to be unweighted, with each clause possibly appearing multiple times. In our formulation, the weight of a clause corresponds to how many times it would appear in $F$ in the unweighted formulation. We use the weighted formulation for convenience. Note however that the weighted formulation is a more efficient method of expressing a formula. For a problem using the unweighted formulation, the input size will in general be larger than for the equivalent instance of the problem using weighted formulation.
This is because rather than encoding a clause together with an integer  $w$, we have to encode the same clause $w$ times.
Therefore when we obtain an algorithm which is polynomial in the input size for the weighted formulation, this is a stronger
result than if we had an algorithm which is polynomial in the input size for the unweighted version.}
\end{remark}

\section{Preliminaries}\label{sec:n}

For a clause $C$, we let $V(C)$ be the set of variables such that $x \in V(C)$ if $x \in C$ or $\bar{x} \in C$.
We assume that every clause $C$ appears only once in $F$.
If at any stage we have two clauses $C_1,C_2$ containing exactly the same literals, we remove one of them, say $C_2$, and add the weight $w(C_2)$ to
$w(C_1)$. In what follows, we will make the following assumption, without loss of generality, for a $3$-satisfiable CNF formula $F$.

\begin{assump}\label{assump1}
All unit clauses in $F$ are of the form $\{x\}$, where $x\in V(F)$.
\end{assump}

Indeed, suppose that $\{\bar{x}\} \in F$, then $\{x\}\notin F$ as $F$ is $3$-satisfiable.
Thus, we may replace $\bar{x}$ by $x$ and $x$ by $\bar{x}$ in all clauses of $F$ without changing ${\rm sat}(F)$.
%
%

\begin{definition}\label{def:partsOfF}
Let $F$ be a $3$-satisfiable CNF formula. We partition $F$ and the variable set $V(F)$ as follows.
\begin{description}
   \item[$F_1$] denotes the set of unit clauses of $F$.
   \item[$\unitvar := V(F_1)$] denotes the set of all variables appearing in unit clauses, called {\em unit variables}.
   \item[$F_2$] denotes the set of all clauses of the form
        $\{\bar{x}, y\}$ or $\{\bar{x}, \bar{y}\}$, where
        $x \in \unitvar$ and $y \not \in \unitvar$.
   \item[$\twovar := V(F_2) \setminus \unitvar$] is the set of non-unit variables in $F_2$.
   \item[$\hard := F_1 \cup F_2$] is the set of {\em hard clauses}.
   \item[$\soft := F \setminus \hard$] is the set of {\em soft clauses}, those that are not hard.
   \item[$\softvar := V(F) \setminus (V_1 \cup V_2)$] are the variables not appearing in any hard clause.
\end{description}
\end{definition}

\begin{definition}\label{def:fat} We say that $F$ is {\em fat} if
$w(\soft) \geq \frac{18}{133} (\hardvars)$. We say that $F$
is {\em hard} if $F = \hard$.
\end{definition}

\begin{remark}
{\rm The hard clauses are called hard because they make it difficult to strengthen the bound $\sat(F) \geq \frac{2}{3}w(F)$. Indeed, if the weight of  soft clauses is significant, one can improve it: in section \ref{fat_proof} we will show that $\sat(F) \geq \frac{2}{3}w(F) + \frac{1}{27} w(\soft)$. The main technical part of this paper deals with proving a better lower bound when $F$ is hard.}
\end{remark}

Let $F$ be a CNF formula.
If $F'$ is a subset of $F$ then $F\setminus F'$ denotes the formula obtained from $F$ by deleting all clauses of $F'$. Let $X$ be a subset of the variables of $F$. Recall that $F_X$ denotes the subset $F$ consisting of all clauses containing a variable from $X.$ Also recall that a CNF formula $F$ is {\em expanding} if $|X|\le w(F_X)$ for each $X\subseteq V(F)$.

A {\em truth assignment} is a function $\alpha : V(F) \rightarrow \{ \textsc{true, false} \}$.
A truth assignment $\alpha$ \emph{satisfies} a clause $C$ if there exists $x \in V(F)$ such that $x \in C$ and
$\alpha(x)=$ \textsc{true}, or $\bar{x} \in C$ and $\alpha(x)=$  \textsc{false}.
 We will denote, by ${\rm sat}_\alpha(F)$, the sum of the weights of clauses in $F$ satisfied by $\alpha$. We denote the maximum value of ${\rm sat}_\alpha(F)$ over all $\alpha$ by ${\rm sat}(F)$.

A function $\beta :\ U \rightarrow \{ \textsc{true, false} \}$, where $U$ is a subset of $V(F)$, is called a {\em partial truth assignment}.
A partial truth assignment $\beta : U \rightarrow \{ \textsc{true, false} \}$ is an {\em autarky} if $\beta$ satisfies all clauses of $F_U$.
Autarkies are of interest, in particular, due to the following simple fact. 

\begin{lemma}\label{lem:aut}\cite{CroGutJonYeo}
 Let $\beta : U \rightarrow \{ \textsc{true, false} \}$ be an autarky for a CNF formula $F$ and let $\gamma$ be any truth assignment on
 $V(F)\setminus U$.
 Then for the combined assignment $\tau := \beta\gamma$,
 it holds that
 ${\rm sat}_{\tau}(F)=w(F_U)+{\rm sat}_{\gamma}(F\setminus F_U)$.
 Clearly, $\tau$ can be constructed in polynomial time given $\beta$ and $\gamma$.
\end{lemma}

 A version of Lemma \ref{lem:aut} can be traced back to Monien and Speckenmeyer \cite{MS1985}.  Autarkies were first introduced in \cite{MS1985}; they
are the subject of much study, see, e.g., \cite{FKS2002}, \cite{Kul03}, \cite{Sze2004}, and see \cite{BunKul09} for an overview.
In this paper we only make use of a small part of the research on autarkies, as we may limit ourselves to the concept of matching autarkies for our proofs.

A \emph{parameterized problem} is a subset $L\subseteq \Sigma^* \times
\mathbb{N}$ over a finite alphabet $\Sigma$. $L$ is
\emph{fixed-parameter tractable} if the membership of an instance
$(I,k)$ in $\Sigma^* \times \mathbb{N}$ can be decided in time
$f(k)|I|^{O(1)},$ where $f$ is a function of the
{\em parameter} $k$ only~\cite{DowneyFellows99,FlumGrohe06,Niedermeier06}.
Given a parameterized problem $L$,
a \emph{kernelization of $L$} is a polynomial-time
algorithm that maps an instance $(x,k)$ to an instance $(x',k')$ (the
\emph{kernel}) such that (i)~$(x,k)\in L$ if and only if
$(x',k')\in L$, (ii)~ $k'\leq h(k)$, and (iii)~$|x'|\leq g(k)$ for some
functions $h$ and $g$.
It is well-known \cite{DowneyFellows99,FlumGrohe06,Niedermeier06} that a decidable parameterized problem $L$ is fixed-parameter
tractable if and only if it has a kernel. By replacing  Condition (ii) in the definition of a kernel  by $k'\le k$,
we obtain a definition of a {\em proper kernel} (sometimes, it is called a {\em strong kernel}); cf. \cite{AF,CFM}.

\section{Main Results}\label{sec:mr}

 Our aim is to prove a lower bound on ${\rm sat}(F)$ that  includes a multiple of the number of variables as a term. It is clear that for general $3$-satisfiable $F$ such a bound is impossible. Indeed, consider a formula containing a single clause $C$ containing a large number of variables. We can arbitrarily increase the number of variables in the formula, and the maximum number of satisfiable clauses will always be 1.
We therefore need a reduction rule that cuts out `excess' variables. Our reduction rule is based on the following lemma proved by Fleischner et al. \cite{FKS2002} (Lemma 10),  Kullmann \cite{Kul03} (Lemma 7.7) and Szeider \cite{Sze2004} (Lemma 9).

\begin{lemma}\label{lem:red}
Let $F$ be a CNF formula and let ${\cal C}(F)$ be a multiset of clauses of $F$ where every clause $C$ appears $w(C)$ times.
Define a bipartite graph, $B_F$, associated with $F$ as follows:  $V(F)$ and ${\cal C}(F)$ are partite sets of $B_F$ and there is an edge between
$v \in V(F)$ and $C \in {\cal C}(F)$ in $B_F$ if and only if $v \in V(C)$.
Given a maximum matching in $B_F$, in time $O(|F|)$ we can find an autarky $\beta : U \rightarrow \{ \textsc{true, false} \}$ such that
$F \setminus F_U$ is expanding.
\end{lemma}

 The papers \cite{FKS2002}, \cite{Kul03} and \cite{Sze2004} actually show that $F \setminus F_U$ is $1$-expanding (see \cite{FKS2002} or \cite{Sze2004} for a definition), which is a slightly stronger result. For our results it is enough that $F \setminus F_U$ is expanding.
An autarky found by the algorithm of Lemma \ref{lem:red} is of a special kind, called a matching autarky; such autarkies were used first by Aharoni and Linial \cite{AL1986}. Note that the autarky found in Lemma \ref{lem:red} can be empty, i.e., $U=\emptyset$.
Lemmas \ref{lem:aut} and \ref{lem:red} immediately imply the following:

\begin{lemma} \label{fm_res}
Let $F$ be a CNF formula and let $\beta : U \rightarrow \{ \textsc{true, false} \}$ be an autarky found by the algorithm of Lemma \ref{lem:red}.
 Then given any truth assignment $\gamma$ on $V(F)\setminus U$, we can find, in polynomial time, a truth assignment $\tau$ such that
 ${\rm sat}_{\tau}(F)=w(F_U)+{\rm sat}_{\gamma}(F\setminus F_U)$,
and $F\setminus F_U$ is an expanding formula.
\end{lemma}

%
%


The following theorem is the main bound of this paper, and the next two sections are dedicated to proving it.

\begin{theorem} \label{main_bound}
Let $F$ be an expanding $3$-satisfiable CNF formula. Then there exists a constant $\rho>0.0044$ such that
${\rm sat}_{\tau}(F) \geq \frac{2}{3} w(F) + \rho |V(F)|$  for some truth assignment $\tau$ that can be found in polynomial time.
\end{theorem}

Theorem \ref{main_bound} follows from the next two propositions, proved in Sections \ref{fat_proof} and \ref{main_proof}, respectively.

\begin{proposition} \label{fat_bound}
 Let $F$ be a fat expanding $3$-satisfiable CNF formula. Then there exists a constant $\rho \ge \frac{2}{453}$ such that
${\rm sat}_{\tau}(F) \geq \frac{2}{3} w(F) + \rho |V(F)|$  for some truth assignment $\tau$ that can be found in polynomial time.
\end{proposition}

\begin{proposition} \label{lean_bound}
 Let $F$ be an expanding $3$-satisfiable CNF formula which is not fat. Then there exists a constant $\rho \ge \frac{2}{453}$ such that
${\rm sat}_{\tau}(F) \geq \frac{2}{3} w(F) + \rho |V(F)|$  for some truth assignment $\tau$ that can be found in polynomial time.
\end{proposition}

%
%
%


As a direct consequence of Theorem \ref{main_bound} and Lemma \ref{fm_res}, we also have the following bound on
${\rm sat}(F)$ for {\em any} $3$-satisfiable CNF formula $F$.

\begin{corollary}\label{cor:bound}
Let $F$ be a $3$-satisfiable CNF formula. Then, in time $O(|F|)$ we can find an autarky $\beta : U \rightarrow \{ \textsc{true, false} \}$ such that $F\setminus F_U$ is expanding. Moreover, there exists a constant $\rho>0.0044$ such that
$${\rm sat}_{\tau}(F) \ge \frac{2}{3}w(F) + \frac{1}{3}w(F_U) + \rho |V(F\setminus F_U)|$$ for some truth assignment $\tau$ that can be found in polynomial time.
\end{corollary}

\begin{corollary}
{\sc 3-S-MaxSat-AE} is fixed-parameter tractable. Moreover, it has a proper kernel with $O(k)$ variables.
\end{corollary}
\begin{proof}
Let $F$ be a $3$-satisfiable CNF formula, let $\beta : U \rightarrow \{ \textsc{true, false} \}$ be an autarky found by the algorithm of Lemma \ref{lem:red} and let $F'=F\setminus F_U.$  We are to decide whether ${\rm sat}(F) \ge \frac{2}{3}w(F)+k,$ where $k$ (an integer) is the parameter.

By  Lemma \ref{fm_res}, ${\rm sat}(F)=w(F_U)+{\rm sat}(F').$ Thus, ${\rm sat}(F) \ge \frac{2}{3}w(F)+k$ if and only if ${\rm sat}(F') \ge \frac{2}{3}w(F')+k',$ where $k'=\lceil \frac{3k-w(F_U)}{3} \rceil.$ Since $F'$ is an expanding $3$-satisfiable formula,  by Theorem \ref{main_bound} we have
${\rm sat}_{\tau}(F') \geq \frac{2}{3} w(F') + \rho |V(F')|$ for some truth assignment $\tau$ that can be found in polynomial time, where $\rho >0.0044$. Thus, if $\rho |V(F')|\ge k',$ then the answer to {\sc 3-S-MaxSat-AE} is {\sc yes} and the corresponding truth assignment can be found in polynomial time. Otherwise, $|V(F')|<\frac{k'}{\rho}$ and, thus, $|V(F')|=O(k)$, and so we can find the optimal assignment in time $2^{O(k)}m^{O(1)}$, where $m = |F|$.

Let $m'=|F'|.$ If $m'\ge 2^{|V(F')|}$, we can find ${\rm sat}(F')$ and, thus, ${\rm sat}(F)$ in polynomial time. Therefore, we may assume that
$m'<2^{|V(F')|}$ and, thus, $m'=2^{O(k)}$ implying that $F'$ is a kernel. Since $k'\le k$, $F'$ is a proper kernel.

\end{proof}

\section{Proof of Proposition \ref{fat_bound}} \label{fat_proof}

The following result is an easy extension of the $\frac{2}{3}w(F)$ bound on ${\rm sat}(F)$. The proof is almost exactly the same as Yannakakis's proof in \cite{Yannakakis94}; in particular the probability distribution involved is the same. The only difference is that our proof involves extra analysis to get the addition of $\frac{1}{27} w(\soft)$.

\begin{lemma} \label{nonhard}
Let $F$ be a $3$-satisfiable CNF formula. Then we can find, in polynomial time, a truth assignment $\tau$ such that
${\rm sat}_{\tau}(F)\ge \frac{2}{3}w(F) + \frac{1}{27} w(\soft).$
\end{lemma}

\begin{proof}
We will construct a random truth assignment $\alpha$ such that $\mathbb{E}({\rm sat}_\alpha(F)) \ge \frac{2}{3}w(F) + \frac{1}{27} w(F_s)$.
This implies that there exists an assignment which satisfies clauses of total weight at least $\frac{2}{3}w(F) + \frac{1}{27} w(F_s)$; we can find such an assignment
in polynomial time using the well-known method of conditional expectations, see, e.g., \cite{alon}.

We define a random truth assignment $\alpha$ as follows. For $x \in \unitvar$, we let $\alpha(x)$ be $\textsc{true}$ with probability $\frac{2}{3}$. For $y \in V(F) \backslash \unitvar$, we let $\alpha(y)$ be $\textsc{true}$ with probability $\frac{1}{2}$. The values are assigned to the variables independently from each other.

Let $C$ be a clause and let $\alpha$ be the random truth assignment above. We will now bound $\mathbb{E}({\rm sat}_\alpha(C))$.
We first consider a hard clause $C$ and, to simplify notation, assume that $w(C)=1$.  By Assumption \ref{assump1}, we have the following cases.

\begin{description}
\item[$C = \{x\}$]: In this case the probability that $C$ is satisfied is exactly $\frac{2}{3}$ and, thus, $\mathbb{E}({\rm sat}_\alpha(C))=\frac{2}{3}.$
\item[$C = \{\bar{x}, y\}$ or $C = \{\bar{x}, \bar{y}\}$ for $x \in \unitvar, y \notin \unitvar$]: Then $\mathbb{E}({\rm sat}_\alpha(C)) = 1 - \frac{2}{3}\times \frac{1}{2} = \frac{2}{3}$.
\end{description}

Thus, for every hard clause $C$  with $w(C)\ge 1$, we have $\mathbb{E}({\rm sat}_\alpha(C))\ge \frac{2}{3}w(C)$. We will now consider
a non-hard clause $C$  and, to simplify notation, assume that $w(C)=1$. The following cases cover all possibilities.

\begin{description}
\item[$|C|=2$ and $|V(C) \cap \unitvar|= 2$]: Let $x_1,x_2\in V_1$.
Observe that $C = \{\bar{x}_1, \bar{x}_2\}$ is not in $F$ as $F$ is $3$-satisfiable and
we cannot satisfy the three clauses $\{x_1\}$, $\{x_2\}$ and $\{\bar{x}_1, \bar{x}_2\}$
simultaneously.
Therefore $\mathbb{E}({\rm sat}_\alpha(C)) \ge 1 - \frac{1}{3}\times
\frac{2}{3} = \frac{7}{9}$.
\item[$|C|=2$ and $|V(C) \cap \unitvar|= 1$]: Then $\mathbb{E}({\rm
sat}_\alpha(C)) = 1 - \frac{1}{3}\times \frac{1}{2} = \frac{5}{6}$.
\item[$|C|=2$ and $|V(C) \cap \unitvar|= 0$]: Then $\mathbb{E}({\rm sat}_\alpha(C)) = 1 - \frac{1}{2}\times \frac{1}{2} = \frac{3}{4}$.
\item[$|C|\ge 3$]: Since for each literal the probability of it being assigned $\textsc{false}$ is at most $\frac{2}{3}$, we have
$\mathbb{E}({\rm sat}_\alpha(C)) \ge 1- (\frac{2}{3})^3 = \frac{19}{27}$.
\end{description}

Thus, for every non-hard clause $C$  with weight $w(C)$, we have $\mathbb{E}({\rm sat}_\alpha(C))\ge \frac{19}{27}w(C)$.
Therefore,
$$\mathbb{E}({\rm sat}_\alpha(F)) \ge \frac{2}{3} w(\hard) + \frac{19}{27} w(F_s) = \frac{2}{3} w(F) + \frac{1}{27} w(F_s).$$
\end{proof}
%
%


Now let $F$ be a fat expanding $3$-satisfiable CNF formula.
We can efficiently find an assignment $\tau$ such that

\begin{eqnarray*}
 {\rm sat}_{\tau}(F) & \ge & \frac{2}{3}w(F) + \frac{1}{27} w(\soft) \mbox{    (by Lemma \ref{nonhard})}\\
   & = &	\frac{2}{3}w(F)+ \frac{133}{27 \cdot 151}w(\soft) + \frac{18}{27 \cdot 151}w(\soft) \\
  & \ge &		\frac{2}{3}w(F)+ \frac{133 \cdot 18}{27 \cdot 151 \cdot 133}(\hardvars) + \frac{18}{27 \cdot 151}|\softvar| \\
  & & \mbox{  (by definitions of fat clauses and an expanding formula)}\\
  & = &		\frac{2}{3}w(F)+ \frac{18}{27 \cdot 151}(\hardvars + |\softvar|) \\
	& = &	\frac{2}{3}w(F)+ \frac{2}{453}|V|.
\end{eqnarray*}

This completes the proof of Proposition \ref{fat_bound}.

\section{Proof of Proposition \ref{lean_bound}} \label{main_proof}

%
%
%
%
%

We will prove Proposition \ref{lean_bound} in the end of this section using Lemma \ref{theorem-lean}.
Lemma \ref{theorem-lean} will be shown using the next two lemmas.
We prove Lemma \ref{lem:nbound} first, as it is somewhat simpler. Note that
the assignments whose existence is claimed in the lemmas can be found efficiently using
the method of conditional expectations mentioned above.

In what follows, assume $F$ is a hard formula, and
let $n_1 := |\unitvar|$, $n_2 := |\twovar|$ and write
$\unitvar = \{x_1,\dots,x_{n_1}\}$ and $\twovar = \{y_1,\dots,y_{n_2}\}$.
For a possibly partial truth assignment $\alpha$ and a formula $F$, we
denote by $\sat_\alpha(F)$ the total weight of the clauses of $F$ that are
satisfied by $\alpha$. Similarly, $\unsat_\alpha(F)$ is the total weight of
the unsatisfied clauses.  Clearly $w(F) = \sat_\alpha(F) +
\unsat_\alpha(F)$. Recall that $F$ consists only of hard clauses, so
$F = \hard = F_1 \cup F_2$. All clauses in $F_1$ are of the
form $\{x_i\}$
(by Assumption~\ref{assump1}), and every clause in
$F_2$ is  of the form either $\{\bar{x}_i, y_j\}$ or
$\{\bar{x}_i, \bar{y}_j\}$.


\begin{lemma}\label{lem:nbound}
  There is an assignment satisfying at least a total weight of
  $\frac{2}{3} w(F) + \frac{2}{9} n_2$.
  \label{lemma-n}
\end{lemma}
\begin{proof}

The main idea of the proof is as follows:
If we set each variable in $\unitvar$ to \textsc{true} with probability $2/3$, the formula $F$ reduces to a $1$-CNF formula $F'$ over $\twovar$. Consider a variable $y \in \twovar$. The weights of $\{y\}$ and $\{\bar{y}\}$ in $F'$ are now random variables taking integer values. Since $F$ cannot contain both $\{\bar{x}, y\}$ and $\{\bar{x}, \bar{y}\}$, those random variables are independent and thus, with a certain constant probability, differ by at least $1$.
Therefore, by setting variables in $\twovar$ optimally, rather than uniformly at random,
we can satisfy more than half the weight of $F'$ and so satisfy more than $\frac{2}{3}w(F)$ clauses (by weight) overall.

We now make this intuition formal.
Set each $x \in V_1$ independently to \textsc{true} with probability $2/3$,
and denote this partial assignment by $\alpha$.
Let $F_2^\alpha$ be the set of
clauses of $F_2$ that are not satisfied by $\alpha$. Thus
$w(F_2^\alpha) = \unsat_\alpha(F_2)$ and 
\begin{equation}\label{eq:unsat} \E_{\alpha} [w(F_2^\alpha)] = \frac{2}{3} w(F_2).\end{equation}



Before assigning values to $\twovar$, let us examine $F_2^\alpha$.
For each variable $y_j \in \twovar$ define the two
random variables
\begin{eqnarray}
  Y^+_j & := & w( \{\bar{x}_i,y_j\} \in F \ | \ \alpha(x_i)=\textsc{true},
  	1 \leq i \leq n_1\}) \\
  Y^-_j & := & w( \{\bar{x}_i,\bar{y}_j\} \in F \ | \ \alpha(x_i)=\textsc{true},
  	1 \leq i \leq n_1\}) \ .	
\end{eqnarray}
Note that $\sum_{j=1}^{n_2} (Y^+_j + Y^-_j) = w(F_2^\alpha)$. By setting $y_j$ to \textsc{true} with probability $1/2$, we could satisfy $(Y^+_j + Y^-_j)/2$. By setting $y_j$ optimally, we can satisfy $\max(Y^+_j, Y^-_j)$, which is possibly more.
In order to estimate the difference, consider the distribution of $Y^+_j - Y^-_j$.  Let $C_1,\dots,
C_\ell$ be the clauses of $F$ containing $y_j$ or $\bar{y}_j$. Define $a_i :=
w(C_i)$ if $y_j \in C_i$ and $a_i := -w(C_i)$ if $\bar{y}_j \in
C_i$. Then
\begin{eqnarray}
Y^+_j - Y^-_j = a_1 z_1 + a_2 z_2 + \dots + a_\ell z_\ell,
\label{eq-Y}
\end{eqnarray}
where $\ell \geq 1$ and the $z_i$ are independent Bernoulli variables
with expectation $2/3$.
To see that they are independent, observe that
for every variable
$x_i$, there is at most one clause $C$ containing both $x_i$ and $y_j$
as variables, namely at most one of $\{\bar{x}_i, y_j\}$ and
$\{\bar{x}_i,\bar{y}_j\}$, by $F$ being $3$-satisfiable. Therefore,
for the clauses containing $y_j$ or $\bar{y}_j$, the events that their
weights contribute to the sum in (\ref{eq-Y}) are independent. Since
$|a_i| \geq 1$ for $1 \leq i \leq \ell$, it is easy to see that with
probability at least $4/9$, the random variable $Y^+_j - Y^-_j$ is
non-zero (the case $\ell=2$, $a_1 = 1, a_2 = -1$ shows that this is
tight).  Therefore,
$$
\E[ |Y^+_j - Y^-_j| ] \geq \frac{4}{9} \ .
$$

We may now give a partial assignment $\beta: \{y_1,\dots,y_{n_2}\} \rightarrow \{\textsc{false},\textsc{true}\}$
based on $\alpha$.
After sampling $\alpha$, we do not sample $\beta$
randomly, but choose each $\beta(y_j)$ optimally: If $Y^+_j - Y^-_j
\geq 0$, set $y_j$ to \textsc{true}, if $Y^+_j - Y^-_j < 0$, set it to
\textsc{false}. Thus, we see that
\begin{eqnarray*}
\sat_\beta(F_2^\alpha) & = & \sum_{j=1}^{n_2} \max(Y^+_j, Y^-_j) \\
& = & \sum_{j=1}^{n_2} \frac{Y^+_j+ Y^-_j}{2} + \frac{|Y^+_j - Y^-_j|}{2}\\
& = & \frac{1}{2} w(F_2^\alpha) + \frac{1}{2}\sum_{j=1}^{n_2}|Y^+_j - Y^-_j| \ .
\end{eqnarray*}

Thus the expected weight of satisfied clauses is
\begin{eqnarray*}
\E_{\alpha} [\sat_{\alpha\beta}(F)] & = &
\E_\alpha[\sat_\alpha(F_1)] + \E_\alpha[\sat_\alpha(F_2)]
+ \E_\alpha\left[\sat_{\beta} (F_2^\alpha)\right] \nonumber \\
& = &
\frac{2}{3}w(F_1) + \frac{1}{3} w(F_2) +
\E_\alpha\left[\frac{1}{2}w(F_2^\alpha) +
\frac{1}{2}\sum_{j=1}^{n_2}|Y^+_j - Y^-_j|\right]
\nonumber \\
& = &
\frac{2}{3}w(F) + \frac{1}{2} \sum_{j=1}^{n_2}\E[|Y^+_j - Y^-_j|] \mbox{    (by (\ref{eq:unsat}))}\\
& \geq &
\frac{2}{3}w(F) + \frac{1}{2} \sum_{j=1}^{n_2} \frac{4}{9}\\
& = & \frac{2}{3}w(F) + \frac{2}{9} n_2 \ .
\end{eqnarray*}
Thus, there is some assignment $\alpha\beta$ satisfying a weight
of at least $\frac{2}{3}w(F) + \frac{2}{9}n_2$.
\end{proof}

\begin{lemma}\label{lem:mbound}
  There is an assignment satisfying at least a total weight of
  $\frac{2}{3} w(F) + \frac{1}{6} n_1$.
  \label{lemma-m}
\end{lemma}
\begin{proof}
This case is almost symmetric to the one above. In a first step,
we sample $\beta$ uniformly at random. Then, instead of
sampling $\alpha$ according to a Bernoulli distribution with
probability $2/3$, we again choose $\alpha$ optimally.

Set each $y \in V_2$ independently to \textsc{true} with probability $1/2$,
and denote this partial assignment by $\beta$.
Now let $F_2^\beta$ be the set of clauses of $F_2$ not satisfied by
$\beta$, and let $F^\beta = F_1 \cup F_2^\beta$. Note that 
\begin{equation}\label{eq:exp}\E_{\beta}[w(F_2^{\beta})]=\frac{1}{2}w(F_2).\end{equation} For $1 \leq i \leq n_1$,
define
$$
Z_i := w\left(\left\{ \{\bar{x}_i, y_j\} \in F_2 \ | \ \beta(y_j)=\textsc{false} \right\}
\cup  \left\{\{\bar{x}_i, \bar{y}_j\} \in F_2 \ | \ \beta(y_j)=\textsc{true} \right\}
\right)
$$
Then $Z_i$ are random variables depending on $\beta$. With this
notation, $w(F_2^\beta) = \sum_{i=1}^{n_1} Z_i$. Now we define $\alpha$ as
follows: If $w(\{x_i\}) \geq Z_i$, set $\alpha(x_i)=\textsc{true}$. If
$w(\{x_i\}) < Z_i$, set $\alpha(x_i)=\textsc{false}$. To bound the expected
satisfied weight, we use the inequality
$$
\max(a,b) \geq \frac{2}{3} a + \frac{1}{3}b + \frac{|a-b|}{3} \ .
$$
This is an equality if $a \geq b$. Now we have
\begin{eqnarray*}
 \E_{\alpha\beta} [\sat_{\alpha,\beta}(F)] & = & \frac{1}{2}w(F_2)
 + \E_{\beta}\left[\sum_{i=1}^{n_1} \max(w(\{x_i\}), Z_i)\right] \\
 & \geq & \frac{1}{2}w(F_2)
 + \E_{\beta}\left[\sum_{i=1}^{n_1} (\frac{2}{3}w(\{x_i\}) + \frac{1}{3}Z_i)
 + \frac{|w(\{x_i\}) -  Z_i|}{3}\right]  \\
 & = & \frac{2}{3} w(F) + \E_{\beta}\left[\sum_{i=1}^{n_1} \frac{|w(\{x_i\}) -
Z_i|}{3}\right] \mbox{    (by (\ref{eq:exp}))}
 \ .
\end{eqnarray*}
Similar to the previous case, observe that $Z_i$ can be written as
$$
b_1 z_1 + \dots + b_k z_k \ ,
$$ where $k \geq 0$, the $b_i$ are positive integers,
and the $z_k$ are independent Bernoulli variables with expectation
$1/2$. Therefore, one may observe that the $w(\{x_i\}) - Z_i$ is non-zero
with probability at least $1/2$. Therefore
\begin{eqnarray*}
   \E_{\alpha\beta} [\sat_{\alpha\beta}(F)] \geq
   \frac{2}{3} w(F) + \frac{1}{6} n_1 \ .
\end{eqnarray*}

Thus, there is some assignment $\alpha\beta$ satisfying a weight
of at least $\frac{2}{3}w(F) + \frac{1}{6}n_1$.
\end{proof}


%



\begin{lemma}\label{theorem-lean}
Let $F$ be a hard $3$-satisfiable formula. Then there is an assignment
satisfying a weight of at least $\frac{2}{3} w(F) + \frac{2}{21} |V(F)|$.
\end{lemma}
\begin{proof}
  Let $n_1,n_2$ be as before. Using Lemmas \ref{lem:nbound} and \ref{lem:mbound}, we can satisfy a total weight of at
  least $\frac{2}{3}w(F) + \max (\frac{1}{6}n_1, \frac{2}{9}n_2)$.
  As the maximum is not smaller than any convex combination, we have
  $$
  \max (\frac{1}{6}n_1, \frac{2}{9}n_2)
  \geq \frac{12}{21} \cdot \frac{1}{6} n_1 + \frac{9}{21} \cdot \frac{2}{9} n_2
   = \frac{2}{21} (n_1+ n_2) \ .
   $$
\end{proof}






\medskip 

It remains to show how Proposition \ref{lean_bound} follows from Lemma \ref{theorem-lean}.
Let $F$ be an expanding $3$-satisfiable CNF formula which is not fat.
Observe that the subformula $F_h$ is a hard formula.
Then by Lemma \ref{theorem-lean} and the definition of an expanding formula,
we can efficiently find an assignment $\tau$ such that

\begin{eqnarray*}
 {\rm sat}_{\tau}(F) & \ge & \frac{2}{3} w(\hard) + \frac{2}{21}(\hardvars)\\
 & = & \frac{2}{3} w(F) - \frac{2}{3} w(\soft) + \frac{2}{21}(\hardvars)\\
& = & \frac{2}{3} w(F) - \frac{2}{3} w(\soft) + \frac{2\cdot 144}{21\cdot 151}(\hardvars) + \frac{2\cdot 7}{21\cdot 151}(\hardvars)\\
& \ge & \frac{2}{3} w(F) - \frac{2}{3} w(\soft) + \frac{2\cdot 144 \cdot 133}{21\cdot 151 \cdot 18}w(\soft) + \frac{14}{21\cdot 151}(\hardvars)\\
& = & \frac{2}{3} w(F) - \frac{14\cdot151}{21 \cdot 151} w(\soft) + \frac{2 \cdot 8 \cdot 133}{21\cdot 151}w(\soft) + \frac{14}{21\cdot 151}(\hardvars)\\
& \ge & \frac{2}{3} w(F) + \frac{(- 2114 + 2128)}{21\cdot 151}|\softvar| + \frac{2\cdot 7}{21\cdot 151}(\hardvars)\\
& = & \frac{2}{3} w(F) + \frac{14}{21 \cdot 151}|\softvar| + \frac{14}{21 \cdot 151}(\hardvars)\\
	& = &	\frac{2}{3} w(F) + \frac{2}{453}|V|.
\end{eqnarray*}

This completes the proof of Proposition \ref{lean_bound}.

\section{Discussion}\label{sec:d}

Let $r_t$ be the largest real such that in any $t$-satisfiable CNF formula at least $r_t$-th fraction of its clauses can be satisfied simultaneously. Note that $r_1=\frac{1}{2}$, $r_2=\frac{\sqrt{5}-1}{2}$ and $r_3=\frac{2}{3}.$ Kr{\'a}l \cite{Kral} established the value of $r_4$: $r_4=3/(5+(\frac{3\sqrt{69}-11}{2})^{1/3}-(\frac{3\sqrt{69}+11}{2})^{1/3})\approx 0.6992$.
For general $t$, Huang and Lieberherr \cite{HL85} showed that $\lim_{t\rightarrow \infty} r_t\le 3/4$ and Trevisan \cite{T1997} proved that $\lim_{t\rightarrow \infty} r_t=\frac{3}{4}$ (a different proof of this result was later given by Kr{\'a}l \cite{Kral}).

By definition, for each $t$-satisfiable CNF formula, we have ${\rm sat}(F)\ge r_tw(F)$. For $t=1,2$ this inequality was improved in \cite{CroGutJonYeo}
and for $t=3$ it was improved in this paper. It would be interesting to find a non-trivial improvement for ${\rm sat}(F)\ge r_tw(F)$ for each $t\ge 1$.

For any $t\ge 1$, a parameterized problem {\sc $t$-S-MaxSat-AE} can be defined as follows: given a $r$-satisfiable formula $F$, verify whether ${\rm sat}(F)\ge r_tw(F)+k$, whether $k$ is the parameter. For $t=1,2,3$, it has been shown that {\sc $t$-S-MaxSat-AE} has a kernel with a linear number of variables. It would be interesting to investigate whether this result can be extended to any $t.$

\paragraph{Acknowledgments}
Research of GG and MJ was partially supported by an International Joint grant of Royal Society.
DS acknowledges
  support from the Danish National Research Foundation and The
  National Science Foundation of China (under the grant 61061130540)
  for the Sino-Danish Center for the Theory of Interactive
  Computation, within which his part of this work was performed.

\end{document}